\documentclass[journal]{IEEEtran}

\usepackage{cite}
\ifCLASSINFOpdf
   \usepackage[pdftex]{graphicx}
\else
\fi

\usepackage[cmex10]{amsmath}
\usepackage{algorithmic,algorithm}

\usepackage{array}
\usepackage[caption=false,font=normalsize,labelfont=sf,textfont=sf]{subfig}
\usepackage{url}

\renewenvironment{IEEEbiography}[1]
  {\IEEEbiographynophoto{#1}}
  {\endIEEEbiographynophoto}

\usepackage{xcolor}

\usepackage{amsthm}
\usepackage{amsfonts}
\usepackage{makecell}

\usepackage[T1]{fontenc}
\usepackage{array}
\newcolumntype{P}[1]{>{\centering\arraybackslash}p{#1}}

\usepackage{makecell}
\newcolumntype{x}[1]{>{\centering\arraybackslash}p{#1}}

\usepackage{tikz}

\newtheorem{proposition}{\bf Proposition}

\newtheorem{remark}{Remark}

\newcommand{\ve}[1]{\boldsymbol{#1}}

\newcolumntype{I}{!{\vrule width 1.2pt}}
\makeatletter
\def\hlinewd#1{%
\noalign{\ifnum0=`}\fi\hrule \@height #1 %
\futurelet\reserved@a\@xhline}

\usepackage{tabularx} 

\IEEEoverridecommandlockouts
\usepackage{fancyhdr}
\makeatletter
\let\ps@plain\ps@fancy
\makeatother

\begin{document}

\bstctlcite{IEEEexample:BSTcontrol} 
\title{Game-Theoretic Electric Vehicle Charging Management Resilient to Non-Ideal User Behavior}

\author{Chathurika~P.~Mediwaththe,~\IEEEmembership{Student~Member,~IEEE,}~David~B.~Smith,~\IEEEmembership{Member,~IEEE,}%
\thanks{Manuscript received (date to be filled in by Editor). }
\thanks{C. P. Mediwaththe and D. B. Smith are with Data61 (NICTA), CSIRO,
Eveleigh, NSW 2015, Australia. e-mail: (m.mediwaththe@unsw.edu.au, David.Smith@data61.csiro.au).}%
\thanks{C. P. Mediwaththe is also with the University of New South Wales, Australia.}
\thanks{D. B. Smith is also with the Australian National University, Australia.}%
}

\maketitle
\thispagestyle{fancy} 
\begin{abstract}

In this paper, an electric vehicle (EV) charging competition, among EV aggregators that perform coordinated EV charging, is explored while taking into consideration potential non-ideal actions of the aggregators. In the coordinated EV charging strategy presented in this paper, each aggregator determines EV charging start time and charging energy profiles to minimize overall EV charging energy cost by including consideration of the actions of the neighboring aggregators. The competitive interactions of the aggregators are modeled by developing a two-stage non-cooperative game among the aggregators. The game is then studied under prospect theory to examine the impacts of non-ideal actions of the aggregators in selecting EV charging start times according to subjectively evaluating their opponents' actions. It is shown that the non-cooperative interactions among the aggregators lead to a subgame perfect $\epsilon$-Nash equilibrium when the game is played with either ideal, or non-ideal, actions of the aggregators. A case study presented demonstrates that the benefits of the coordinated EV charging strategy, in terms of energy cost savings and peak-to-average ratio reductions, are significantly resilient to non-ideal actions of the aggregators.

\end{abstract}


\section{Introduction}
Escalating fuel prices and environmental concerns have increased the market penetration of electric vehicles (EVs) worldwide. For example, a recent study \cite{EVdat} has shown that the annual growth rate of EV sales in the United States is more than 20\%. Despite environment-friendly features, uncoordinated EV charging creates challenges to both economical and technical aspects of the power grid as a consequence of excessive load consumption. Due to the escalation in electricity demand provoked by EVs, demand-side management has become a paramount element in power system operation as it can decrease EV charging costs by coordinating the grid-to-vehicle operation economically. EV aggregators have been widely adopted in the present energy market with the aim of facilitating coordinated charging of EVs at large scale. An EV aggregator acts as a middleman, between a fleet of EVs and the power grid, to optimally regulate the charging plan of the vehicle fleet so as to minimize overall cost of EV charging considering EV charging constraints \cite{Aliprantis,ZechunHu}.

Effective demand-side management of grid-to-vehicle operation requires active participation of users who are either vehicle owners or EV aggregators that seek to optimize individual goals through coordinated EV charging. However, in the long run, users might deviate from their ideal participating behavior despite the benefits that they gain through participating in demand-side management. Unpredictable non-ideal behavior of users is likely to compromise the economic benefits and system efficiencies of demand-side management. A successful energy management approach for EV charging is often challenging with inconsistent user behavior \cite{Franke}. 

In this paper, the EV charging competition among multiple EV aggregators in a coordinated EV charging system is studied while accounting for actions of the aggregators that are not completely rational. Each aggregator, which could be run by, e.g., a car park manager, determines EV charging start time and charging energy profiles to minimize EV charging energy cost by considering the actions of the neighboring EV aggregators. The interactions among the EV aggregators are modeled using a non-cooperative game-theoretic framework, which is then studied under two user behavioral models, expected utility theory and prospect theory, to incorporate aggregators' ideal and non-ideal actions in selecting EV charging start times in the EV charging competition. The main contributions of this paper can be stated as follows:
\begin{itemize}
\item We model the coordinated EV charging competition among the EV aggregators as a two-stage non-cooperative game and show that there exists a subgame perfect $\epsilon$-Nash equilibrium when the game is played with either ideal, or non-ideal, actions of the aggregators.

\item We analyze the impacts of non-ideal actions of the EV aggregators through an extensive performance analysis and show that the benefits of the coordinated EV charging strategy, in terms of peak-to-average ratio reductions and EV charging cost savings, are resilient to non-ideal actions taken by the aggregators. 
\end{itemize}  

Prospect theory has emerged as a prominent tool to examine non-ideal user behavior that departs from the rational choices in game theory. To the best of our knowledge, few works related to demand-side management have applied prospect theory to study non-ideal, realistic behavior of participants that cannot be explained by assuming consumer rationality \cite{Wang,ISIE}. For example, \cite{Wang} elaborates insights from prospect theory to study realistic consumer decision-making in a load-shifting demand response program for residential households. In \cite{ISIE}, prospect theory is used to study the user behavior in a Stackelberg game-theoretic energy trading system between a community energy storage device and photovoltaic energy users. Due to the inherent differences in system models and associated constraints between the EV charging scenario and the residential load-shifting and bi-level energy trading systems, the prospect theory-based analyses in \cite{Wang,ISIE} cannot be directly applied to the EV charging scenario considered in this paper. Furthermore, to the best of our knowledge, in literature, the effects of non-ideal actions of participating users on an EV charging scenario have not been investigated using prospect theory. 

The remainder of this paper is organized as follows. Section~\ref{sec:2} presents related work. Section~\ref{sec:3} describes the EV charging system configuration, and Section~\ref{sec:4} describes the two-stage non-cooperative game among the EV aggregators. Section~\ref{sec:5} elaborates the non-cooperative EV charging energy determination game among the aggregators and Section~\ref{sec:6} discusses the participation time selection game of the aggregators. Section~\ref{sec:7} presents numerical results and Section~\ref{sec:8} concludes the paper.

\section{Related Work}\label{sec:2}
Strategic decision-making of users has been widely investigated, by modeling user behavior, in many user-centric applications in sociology, economics, and energy markets \cite{Shen,Harrison,Chrysopoulos}. Game theory has been a popular analytical platform in demand-side management literature to model how users strategically reason about their neighbors' behavior to determine their own energy strategies such that the system-wide objectives are optimized \cite{Fadlullah,Deng,Soliman,Song,Su2014341}. The demand-side management approach in \cite{Mohsen} optimally schedules the user energy consumption through a non-cooperative game among energy users. The non-cooperative Stackelberg game in \cite{Mengmeng} explores optimal load interaction between a utility company and residential energy users to balance energy demand and supply.  Moreover, the insights of game theory have been used to study the energy trading competition between shared energy storage systems and self-interested solar energy users for small-scale demand-side management \cite{Chathurika1, Chathurika2}.

Game theory is also used in the context of distributed optimal operation in EV charging systems \cite{Dextreit,Lee,Bahrami, Malandrino}. Based on game theory, \cite{Lee} explores optimal price competition among EV charging stations with renewable energy generation to attract EVs. Using a non-cooperative game-theoretic approach, \cite{Bahrami} investigates an optimal demand response method for effective scheduling of plug-in hybrid EV charging. By investigating a non-cooperative game among self-interested EV owners, an optimal valley filling pricing mechanism to coordinate EV charging is proposed in \cite{ZechunHu}. Using a non-cooperative game among EVs at a parking lot, \cite{Zhang} studies an effective way of obtaining EV charging strategies at a unique Nash equilibrium under the constrained distribution transformer capacity. The non-cooperative Stackelberg game-theoretic approach in \cite{Tushar} investigates the optimal bi-level grid-to-vehicle coordination between the power grid and plug-in EV groups with limited energy supply. Most of the game-theoretic demand-side management literature optimizes important aspects of various energy markets conforming to conventional game theory that fundamentally assumes rational choices by users. 

Empirical evidence from social studies has illustrated that the axiom of rationality of game theory can be violated when users face risk and uncertainty in decision-making \cite{Kahneman}. In \cite{Sanstad}, discrepancies between rational choice and alternative user behavioral models in an energy market are reviewed. Prospect theory has attracted great attention in various research communities to understand impacts of deviations from rationality \cite{Li,LiTrans,Wang}. A comprehensive discussion on the potential of prospect theory to understand how decision-making under risk and uncertainty can affect various smart grid applications is given in \cite{Glass}. In \cite{Wang2}, optimal energy exchange among geographically-distributed consumer-owned energy storage devices is compared under classical game theory and prospect theory. In \cite{Wang3}, the effects of subjective consumer behavior on an energy storage charging-discharging system are investigated using a utility framing method derived from prospect theory. In contrast to previous work, this paper presents precepts of prospect theory to realize potential non-ideal behavior of EV aggregators in a non-cooperative game-theoretic EV charging system.

\section{System Configuration}\label{sec:3}
This section describes the formulation of the system models of EV aggregators and the cost models that are used to derive energy costs in this paper. 

\subsection{Electric Vehicle Charging Model}\label{sec:3_1}

In this paper, a low-voltage power grid with multiple EV aggregators distributed along the grid is considered. Here, an EV aggregator acts as an intermediary between the power grid and a fleet of EVs and regulates and schedules charging of the connected EVs considering EV charging constraints. In this work, only grid-to-vehicle operation with unidirectional power flow is considered.

In the system model, the set of EV aggregators is denoted as $\mathcal{N}$ and $|\mathcal{N}|=N$. Each aggregator $i\in \mathcal{N}$ controls an EV charging station that consists of multiple EV chargers within a localized geographical area. For example, such charging stations may be located in car parks at workplaces, universities, and shopping centers. Each aggregator $i$ regulates charging of a set of EVs $\mathcal{V}_i$ over a set period of time $\mathcal{T}$ of the day and $|\mathcal{V}_i|=V_i$. In this paper, the EV charging time frame $\mathcal{T}$ is assumed to span from $8.00$~AM to $4.00$~PM considering a workplace charging scenario.  In this case, it is considered that EVs in $\mathcal{V}_i$ arrive at the charging station by $8.00$~AM and park at the charging station for the entire time period of $\mathcal{T}$, committing to an agreement between EV owners and the aggregator $i$. The charging horizon $\mathcal{T}$ is divided into $T$ number of time slots of $\Delta t$ length and the control time is denoted by $t\in \{1,2,\dotsm,T\}$. In this situation, the aggregators, which may be run by car park managers, can coordinate EV charging efficiently to charge EVs within the considered time period $\mathcal{T}$. At the end of time period $\mathcal{T}$, it is required to charge all EVs in $\mathcal{V}_i$ to their maximum state-of-charge (SOC) limits. Furthermore, this work assumes that additional EVs will not join the set $\mathcal{V}_i$ during the time period $\mathcal{T}$ similar to the EV charging framework considered in \cite{Richardson}.

Without loss of generality, it is considered that there is a base load profile $\ve{L}_b$ on the grid that is not time-flexible and $\ve{L}_b$ is given by

\begin{equation}
\ve{L}_b=(L_{b,1},\dotsm,L_{b,t},\dotsm,L_{b,T})\label{eq1}
\end{equation}
where $L_{b,t} (>0)$ is the base grid load at time $t$. For example, $\ve{L}_b$ may constitute non-deferrable loads of users of the grid. From the utility's perspective, $L_{b,t}$ represents the minimum load that the utility should provide at time $t$. 

For each aggregator $i$, there is an energy demand $E_i$ that is equal to the total energy demands of the EVs in $\mathcal{V}_i$. The energy demand of an EV $v\in \mathcal{V}_i$ is given by

\begin{equation}
e_v=S_{v,\text{init}}-S_{v, \text{max}}\label{eq2}
\end{equation}
where $S_{v,\text{init}}$ and $S_{v, \text{max}}$ are the SOC level at the beginning of $\mathcal{T}$ and maximum SOC level of the battery of EV $v$, respectively. Using \eqref{eq2}, $E_i$ can be written as $E_i=\sum_{v=1}^{V_i}e_v$. In the EV charging model, each aggregator $i$ supplies energy $E_i$ to the EVs in $\mathcal{V}_i$ by distributing it across time $\mathcal{T}$. In this situation, the temporal grid energy consumption profile of aggregator $i$ over the time period $\mathcal{T}$ is given by

\begin{equation}
\ve{x}_i=(x_{i,1},\dotsm,x_{i,t},\dotsm,x_{i,T})\label{eq3}
\end{equation}
 where $x_{i,t}$ is the energy amount taken from the grid by aggregator $i$ at time $t$. 
 
To consider conversion losses of EV chargers controlled by aggregator $i$, a charging efficiency parameter $\eta_i$ is introduced such that $0<\eta_i\leq 1$. For example, if $x_{i,t}$ amount of energy is taken from the grid by aggregator $i$, only $\eta_ix_{i,t}$ amount is effectively dispatched for charging EVs in $\mathcal{V}_i$. Given \eqref{eq3} and $\eta_i$, $E_i$ of each aggregator $i$ satisfies
 
 \begin{equation}
E_i=\sum_{t=1}^T \eta_ix_{i,t}.\label{eq4}
\end{equation}

It is considered that each EV $v\in \mathcal{V}_i$ is charged using a charging rate between a maximum charging rate $R_v$ and a minimum charging rate taken as zero at time $t$ \cite{Richardson,Clement}. This gives

 \begin{equation}
0\leq r_{v,t}\leq R_v\label{eq5}
\end{equation}
where $r_{v,t}$ is the charging rate of EV $v\in \mathcal{V}_i$ at time $t$. Considering \eqref{eq5} for all EVs in $\mathcal{V}_i$, energy consumption $x_{i,t}$ satisfies

 \begin{equation}
0 \leq \frac{\eta_ix_{i,t}}{\Delta t} \leq \sum_{v=1}^{V_i}R_v.\label{eq6}
\end{equation}
In this paper, it is assumed that the total power demand of aggregators at each time $t$ can be obtained from the grid without violating the grid voltage and capacity constraints.

\subsection{Energy Cost Models}\label{sec:3_2}

In this paper, a dynamic grid cost function that consists of both real time and time-of-use pricing elements is considered \cite{Mohsen}. Given the total grid load as $L_t=L_{b,t}+\sum_{i=1}^N x_{i,t}$, the grid cost function at time $t$ is given as 

 \begin{equation}
P_t(L_t)=\phi_tL_t^2+\delta_tL_t\label{eq7}
\end{equation}
where $\phi_t$ and $\delta_t$ are positive time-of-use tariff constants at time $t$. The cost function \eqref{eq7} can be regarded as a quadratic function that approximates piecewise linear pricing models adopted by some electric utility companies \cite{Mohsen,FLi}. By incorporating time-of-use and real time pricing components with these cost models, users can be encouraged to shift their peak demand to non-peak hours \cite{Mohsen,Lambotharan,Mengmeng}. According to \eqref{eq7}, per unit electricity price of the grid at time $t$, $p_t$, is given as $p_t=\phi_tL_t+\delta_t$, and the resulting grid energy cost of aggregator $i$ at time $t$ is given by $p_tx_{i,t}$. 

In addition to the grid energy cost, another cost component $D_{i,t}$ is defined for each aggregator $i$ to model their utility based on the deviation between the actual energy consumption and the target energy consumption at time $t$ \cite{Jiang}. Denoting the target energy consumption of aggregator $i$ at time $t$ as $\bar{x}_{i,t}$, $D_{i,t}$ is given by

\begin{equation}
D_{i,t}=\begin{cases}
                g_{i,t} (\bar{x}_{i,t}-{x}_{i,t})^2,~~\text{if}~0\leq {x}_{i,t} < \bar{x}_{i,t}, \\
                 0,~~~~~~~~~~~~~~~~~~~~\text{if}~{x}_{i,t}\geq \bar{x}_{i,t}
                   \end{cases}\label{eq8}
\end{equation}
where $g_{i,t}(>0)$ is a weighting parameter related to aggregator $i$ that measures how the cost component $(\bar{x}_{i,t}-{x}_{i,t})^2$ affects total aggregator cost function. In this paper, the target energy demands $\bar{x}_{i,t}$ of each aggregator $i$ are evaluated such that $\bar{x}_{i,t}$ at time $t$ is equal to the average of the total demand that needs to be supplied within the time frame $[t,t+1,\dotsm,T]$. In particular, $\bar{x}_{i,t}$ is evaluated such that $\bar{x}_{i,t}=(E_i-\sum_{t=t_i}^{(t-1)}x_{i,t})/((T-t)+1);~t_i\leq t \leq T$. The cost term \eqref{eq8} motivates charging of EVs connected to each aggregator $i$ above the average demand at each time $t$. In doing so, it encourages drawing more energy at the beginning of time frame $t_i\leq t \leq T$ so that charge levels of EV batteries can be maintained at reasonable values even if the EVs depart the system prior to the expected departure time.

In this framework, if aggregator $i$ charges EV $v\in \mathcal{V}_i$ using maximum charging rate $R_v$, then they require ${e_v}/{R_v}$ number of time slots to reach $S_{v, \text{max}}$ from $S_{v,\text{init}}$. It is assumed that if aggregator $i$ starts EV charging at time slot $t_i\in \{1,2,\dotsm,T\}$, then they continue the charging process for $t_i\leq t \leq T$. For instance, if aggregator $i$ starts charging the EVs in $\mathcal{V}_i$ at time slot 3, then the aggregator determines $(x_{i,3},~x_{i,4},\dotsm,~x_{i,T})$ such that $E_i=\sum_{t=3}^{T}\eta_ix_{i,t}$. Given this assumption and if aggregator $i$ charges each EV $v\in \mathcal{V}_i$ using their maximum charging rate $R_v$, then the aggregator requires $\tau_i$ time slots to finish charging all EVs in $\mathcal{V}_i$ where

 \begin{equation}
\tau_i=\text{max}(e_1/R_1,~e_2/R_2, \dotsm,~e_{V_i}/R_{V_i}).\label{eq9}
\end{equation}

Since aggregator $i$ can vary the charging rates of each EV in $\mathcal{V}_i$ according to \eqref{eq5}, $\tau_i$ represents the minimum number of time slots that aggregator $i$ requires to charge all EVs in $\mathcal{V}_i$. This implies that for a given $\tau_i$, aggregator should start EV charging at least from time slot $\tilde{t}_i$ where $\tilde{t}_i=(T-\tau_i)+1$. Hence, aggregator $i$ can start EV charging between time slots 1 and $\tilde{t}_i$. Then, the set of all possible EV charging start times $\mathcal{I}_i$ of aggregator $i$ can be written as $\mathcal{I}_i=\{1,2,\dotsm,\tilde{t}_i\}$ where $\mathcal{I}_i\subseteq\{1,2,\dotsm,T\}$. If aggregator $i$ starts the EV charging process at time slot $t_i\in \mathcal{I}_i$, then their total energy cost is given by

 \begin{equation}
C_i=\sum_{t=t_i}^T(p_tx_{i,t}+D_{i,t}).\label{eq10}
\end{equation}

\section{Two-stage Non-cooperative Game}\label{sec:4}

To analyze how each aggregator $i$ determines their EV charging start time $t_i$ and EV charging energy amounts $x_{i,t}$ at each time $t$, a two-stage non-cooperative game $\Upsilon$ among the aggregators $\mathcal{N}$ is developed. At the first stage of the game $\Upsilon$, the aggregators $\mathcal{N}$ non-cooperatively determine their EV charging start time $t_i \in\mathcal{I}_i$ with imperfect information \cite{gametheoryessentials}. After observing how the aggregators have selected $t_i\in \mathcal{I}_i$ at the first stage, each aggregator $i$ non-cooperatively determines $x_{i,t}$ with imperfect information at the second stage.

Since each aggregator $i$ has $|\mathcal{I}_i|$ number of possible actions in total at the first stage of the game $\Upsilon$, the extensive form of the game $\Upsilon$ implies that at the second stage, the game $\Upsilon$ has $K=(|\mathcal{I}_1|\times|\mathcal{I}_2|\times \dotsm \times |\mathcal{I}_N|)$ number of proper subgames \cite{Fudenberg}. It is considered that the set of proper subgames at the second stage are as $\mathcal{G}=\{G_1,~G_2,\dotsm,G_K\}$. Then the game $\Upsilon$ has $(K+1)$ number of subgames including the entire game itself. In this paper, the strategic form of the game $\Upsilon$ is described as follows.

\begin{itemize}
\item{\textit{Players}}: The set of aggregators $\mathcal{N}$.
\item{\textit{Strategies}}: Each aggregator $i$ selects $(t_i,\ve{x}_i)$ to maximize their payoff. In particular, at the first stage, each aggregator $i$ determines $t_i\in \mathcal{I}_i$. Then at the second stage, depending on the selected $t_i$ at the first stage, each aggregator $i$ determines $\ve{x}_i=(x_{i,t_i},x_{i,t_i+1},\dotsm,x_{i,T})$ such that $\ve{x}_i\in \mathcal{X}_i$ where $\mathcal{X}_i$ subject to constraints \eqref{eq4} and \eqref{eq6}.
\item{\textit{Payoffs}}: For an action profile $(t_i,\ve{x}_i)$, aggregator $i$ receives a payoff
\begin{equation}
U_i=-C_i=-\sum_{t=t_i}^T (p_tx_{i,t}+D_{i,t}).\label{eq11}
\end{equation}
\end{itemize}

To determine the solutions of the game $\Upsilon$, first, the optimal solutions for each subgame in $\mathcal{G}$ at the second stage are evaluated. Then the analysis proceeds backwards to the first stage where each aggregator $i$ determines optimal $t_i\in \mathcal{I}_i$ that maximizes their payoffs, which result if the optimal actions determined for each subgame in $\mathcal{G}$ are adopted by the aggregators $\mathcal{N}$ at the second stage. The explicit analyses of these two steps are given in the next two sections.

\section{Second Stage Game: Charging Energy Determination Game}\label{sec:5}

This section explains the process of determining EV charging energy amounts by each aggregator $i$ once they have selected to start EV charging from time slot $t_i\in \mathcal{I}_i$. Let $G_{\sigma}\in \mathcal{G}$ be the subgame among the aggregators $\mathcal{N}$ at the second stage if they adopt an EV charging start time profile $\ve{\sigma}=(t_1,t_2,\dotsm,t_N) \in \mathcal{I}$ at the first stage of the game $\Upsilon$. Here, $\mathcal{I}$ denotes the cartesian product of strategy sets $\mathcal{I}_i$ of the aggregators $\mathcal{N}$ that is given by $\mathcal{I}=\{\mathcal{I}_1\times\mathcal{I}_2\times\dotsm\times\mathcal{I}_N\}$. Because each aggregator $i$ continues EV charging for $t_i\leq t\leq T$ once they have selected to start EV charging at time $t_i$, in the game $G_{\sigma}$, each aggregator $i$ seeks to maximize their individual payoff given in \eqref{eq11}. In this scenario, their individual decisions on $x_{i,t}$ are influenced by each others' energy consumption decisions due to the aggregate load dependency of the grid price $p_t$. The strategic form of the game $G_{\sigma}$ is given as follows.

\begin{itemize}
\item{\textit{Players}}: The set of aggregators $\mathcal{N}$.
\item{\textit{Strategies}}: Each aggregator $i$ determines $\ve{x}_i\in \mathcal{X}_i$ to maximize payoff.
\item{\textit{Payoffs}}: Each aggregator $i$ receives a payoff given by \eqref{eq11}.
\end{itemize}

In the game $G_{\sigma}$, each aggregator $i$ solves the local optimization problem to determine

 \begin{equation}
\tilde{\ve{x}}_i=\underset {\ve{x}_i \in \mathcal{X}_i}{\textrm{argmax}}~U_i(\ve{x}_i,\ve{x}_{-i})\label{eq12}
\end{equation}
where $U_i(\ve{x}_i,\ve{x_{-i}})\equiv U_i$ and $\ve{x}_{-i}$ is the EV charging energy profile of the other aggregators $i'\in \mathcal{N}\backslash i$.

\begin{proposition}
The game $G_{\sigma}$ has a unique pure strategy Nash equilibrium.
\end{proposition}

\begin{proof}
For a given $\ve{x}_{-i}$, the objective function in \eqref{eq12} is strictly concave with respect to $\ve{x}_i$ as its Hessian matrix with respect to $\ve{x}_i$ is negative definite. Moreover, the strategy set $\mathcal{X}_i$ of each aggregator $i$ is non-empty, compact, and convex due to linearity of \eqref{eq4} and \eqref{eq6}. Therefore, the game $G_{\sigma}$ is a concave N-person game and has a pure strategy Nash equilibrium \cite{rosen}. Moreover, that Nash equilibrium is unique as the objective function and the strategy sets in \eqref{eq12} satisfy Theorem 2 in \cite{rosen}.
\end{proof}

The Nash equilibrium of the game $G_{\sigma}$ is denoted by $\ve{x}^*=(\ve{x}_1^*,\ve{x}_2^*,\dotsm,\ve{x}_N^*)$. In this paper, $\ve{x}^*$ is approximated using the iterative best-response algorithm given in Algorithm~\ref{Al1}. The algorithm terminates when the relative distance of $\ve{x}$ between two consecutive iterations is very small, for example, ${\|\ve{x}^{(k)}-\ve{x}^{(k-1)}\|}_2/{\|\ve{x}^{(k)}\|}_2 \leq\varepsilon$ where $\varepsilon$ is a very small positive value and $k$ is the iteration number.

 \begin{remark}
If the charging model allowed random arrivals and departures of EVs during the charging time frame $\mathcal{T}$, the game $G_{\sigma}$ among the aggregators would turn into a game with incomplete information where uncertainty occurs over the strategy spaces $\mathcal{X}_i$ available to each aggregator $i$. 
\end{remark}

 \begin{algorithm}
 \caption{Game to obtain the Nash Equilibrium of the game $G_{\sigma}$}\label{Al1}
 \begin{algorithmic}[1]
 \STATE Using EV charging start time $t_i$, randomly initialize $\ve{x}_i$ for each aggregator $i$ such that $\ve{x}_i \in \mathcal{X}_i$ and set $k\leftarrow0$.
 \WHILE {termination criterion is not satisfied}
\STATE{Set $k \leftarrow k+1$.}
   \FOR {$\text{each aggregator}~i$}
  \STATE Aggregator $i$ solves \eqref{eq12} and determines $\tilde{\ve{x}}_i$ using the temporal aggregate grid load vector, excluding $i$'s load, $\ve{L}^{(k-1)}_{-i}=\Big(L_{t_i,-i}^{(k-1)},\dotsm,L_{t,-i}^{(k-1)},\dotsm,L_{T,-i}^{(k-1)}\Big)$ at the previous iteration $k-1$. Here, $L_{t,-i}^{(k-1)}=L_t^{(k-1)}-x_{i,t}^{(k-1)}$.
  \ENDFOR
  \ENDWHILE
 \RETURN The Nash equilibrium $\ve{x}^*$.
 \end{algorithmic}
 \end{algorithm}
 
 \section{First Stage Game: Participation Time Selection Game}\label{sec:6}
 
In this section, the selection of optimal EV charging start times of the aggregators $\mathcal{N}$ at the first stage of the game $\Upsilon$ is described. Note that, since the solution analysis of the game $\Upsilon$ moves backwards, the game at the first stage has a payoff for its each action profile $\ve{\sigma}\in \mathcal{I}$ equals to the Nash equilibrium payoff obtained for its corresponding subgame $G_{\sigma}\in \mathcal{G}$ at the second stage. Let $\Phi$ denote the non-cooperative game among the aggregators $\mathcal{N}$ at the first stage that has payoffs equivalent to the Nash equilibrium payoffs of subgames $\mathcal{G}$ at the second stage. Explicitly, the game $\Phi$ can be described as follows.
 
\begin{itemize}
\item{\textit{Players}}: The set of aggregators $\mathcal{N}$.
\item{\textit{Strategies}}: Each aggregator $i$ determines $t_i\in \mathcal{I}_i$ to maximize payoff.
\item{\textit{Payoffs}}: If the aggregators' EV charging start time profile is $\ve{\sigma}\in \mathcal{I}$, then each aggregator $i$ receives a payoff given by

\begin{equation}
F_i(\ve{\sigma})=-\sum_{t=t_i}^T p_t^*(\ve{\sigma})x_{i,t}^*(\ve{\sigma})+ D_{i,t}^*(\ve{\sigma})   \label{eq13}
\end{equation}
where $p_t^*(\ve{\sigma})$ and $x_{i,t}^*(\ve{\sigma})$ are the unit grid price and the EV charging energy amount of aggregator $i$ at time $t$ obtained at the Nash equilibrium of the game $G_{\sigma}$, respectively. Furthermore, $D_{i,t}^*(\ve{\sigma})$ is the cost given in \eqref{eq8} after obtaining the Nash equilibrium of the game $G_{\sigma}$. 
\end{itemize}

Let $\ve{\sigma}_{-i}$ denote the EV charging start time strategy profile of all aggregators $\mathcal{N}$ excluding aggregator $i$. Then in the game $\Phi$, each aggregator $i$ maximizes the payoff given in \eqref{eq13} by determining the optimal $t_i \in \mathcal{I}_i$ for a given $\ve{\sigma}_{-i}$.

\begin{proposition}
A Nash equilibrium strategy profile of the game $\Phi$ leads to a subgame perfect Nash equilibrium of the game $\Upsilon$.
\end{proposition}

\begin{proof}
The proof of the proposition immediately follows intuitions of backward induction \cite{Fudenberg}. In particular, when considering the extensive form (tree form) of the two-stage game $\Upsilon$, the game $\Phi$ is the reduced game of the game $\Upsilon$ after eliminating all subgames $\mathcal{G}$ at the second stage by assigning their Nash equilibrium outcomes to the outcomes after the first stage of the game $\Upsilon$. For example, with respect to the strategy profile $\ve{\sigma}\in \mathcal{I}$, the payoff of each aggregator $i$ in \eqref{eq13} is defined by assuming that the aggregators $\mathcal{N}$ will play the Nash equilibrium of the corresponding subgame $G_{\sigma}$ at the second stage. Therefore, a Nash equilibrium of the game $\Phi$ leads to a subgame perfect Nash equilibrium of the game $\Upsilon$.
\end{proof}

\begin{remark}
Once the aggregators $\mathcal{N}$ have determined their optimal EV charging start time profile $\ve{\sigma}^*$ by playing the game $\Phi$, they play the non-cooperative game $G_{\sigma^*}\in \mathcal{G}$ where $G_{\sigma^*}$ is the subgame at the second stage of the game $\Upsilon$ subsequent to $\ve{\sigma}^*$. 
\end{remark}

In the long run, the aggregators $\mathcal{N}$ may change their behavior with respect to selecting EV charging start times $t_i$. Therefore, to determine solutions for the game $\Phi$, aggregators' empirical frequencies of choosing start times $t_i \in \mathcal{I}_i$ are considered under the notion of mixed strategies. In this scenario, the aggregators $\mathcal{N}$ face uncertainty in decision-making with their probabilistic choices of EV charging start times. Hence, in this paper, the strategic behavior of the aggregators $\mathcal{N}$ is studied under two user behavioral models: expected utility theory, i.e., the conventional game-theoretic approach, and prospect theory that learns the subjective non-ideal behavior of users \cite{Kahneman}.

\subsection{Time Selection under Expected Utility Theory}\label{sec:6_1}

To analyze the game $\Phi$ under mixed strategies, it is considered that each aggregator $i$ evaluates the probability distribution over their strategy set $\mathcal{I}_i$ to maximize expected payoff. In classical game theory, expected utility theory is the main platform that is used to describe user behavior and payoffs under the notion of mixed strategy. In this paper, the expected payoff of each aggregator $i$ under expected utility theory is given by

\begin{equation}
Q_{i,\Phi}^{\text{EUT}}(\ve a) = \sum_{\ve{\sigma}\in \mathcal{I}} F_i(\ve{\sigma}) \prod_{j\in \mathcal{N}} a_j(t_j)  \label{eq14}
\end{equation}
where $\ve a =(\ve{a}_i,\ve{a}_{-i}),~\ve{a}_i=(a_i(1),a_i(2),\dotsm,a_i(\tilde{t}_i)),~a_i(t_i)$ is the probability that aggregator $i$ selects time slot $t_i$ as the EV charging start time, and $\ve{a}_{-i}$ is the probabilities of the other aggregators $i'\in \mathcal{N}\backslash i$ of choosing their EV charging start times.

Intuitively speaking, the payoff calculation under expected utility theory implies that players assess probabilities of their opponents' actions identical to their objective likelihoods. However, empirical evidence from sociology infers that this assumption may not be valid in many real world applications as often users, such as car park managers in our case study, underweight high probability events and overweight low probability events when they face risk and uncertainty \cite{Kahneman}.

\subsection{Time Selection under Prospect Theory}\label{sec:6_2}

The intuition behind prospect theory is to describe user behavior that cannot be understood by assuming rational choices of users as in normative expected utility theory \cite{Kahneman}. In the real world, it has been shown that people exhibit subjective behavior rather than objective behavior in payoff maximization problems. This section investigates how the aggregators $\mathcal{N}$ maximize their utilities in the game $\Phi$ while subjectively evaluating their neighbors' behavior. 

Prospect theory uses weighting effects to characterize the subjective behavior of users \cite{Kahneman} and in particular, probability weighting functions are widely investigated \cite{Neilson, prelec}. In general, a probability weighting function $w_i(a)$ indicates the subjective evaluation of aggregator $i$ on an action played with probability $a$. In this paper, Prelec's probability weighting function \cite{prelec} is used that is given by

\begin{equation}
w_i(a)=\text{exp}(-(- \text{ln}~a)^{\alpha_i}) \label{eq15}
\end{equation}
where $\alpha_i$ is a weighting parameter of aggregator $i$ and $0<\alpha_i\leq 1$. It is important to note that aggregator $i$ becomes more subjective and deviates more from objective evaluation over probabilities as $\alpha_i$ moves from $1$ towards $0$. On the other hand, $\alpha_i=1$ implies that aggregator $i$ perceives the action with probability $a$ objectively, and therefore, their subjective evaluation  and objective evaluation are identical. 

Here, it is assumed that the subjective probabilities of aggregator $i$ of their own actions are equal to the objective probabilities. Then the expected payoff of each aggregator $i$ under prospect theory is given by

\begin{equation}
Q_{i,\Phi}^{\text{PT}}(\ve{a})=\sum_{\ve{\sigma}\in \mathcal{I}} F_i(\ve{\sigma})a_i(t_i)\bigg(\prod_{i'\in \mathcal{N}\backslash i} w_i(a_{i'}(t_{i'}))\bigg). \label{eq16}
\end{equation}

\subsection{$\epsilon$-Nash equilibria}\label{sec:6_3}

In this section, the solutions for the game $\Phi$ under expected utility theory and prospect theory are studied. In particular, $\epsilon$-Nash equilibria for the game $\Phi$ are investigated under two models due to their attractive properties such as computational usefulness and the guarantee that every Nash equilibrium is surrounded by $\epsilon$-Nash equilibria for small $\epsilon>0$ \cite{gametheoryessentials}. A mixed strategy profile $\hat{\ve{a}}=(\hat{\ve{a}}_i,\hat{\ve{a}}_{-i})$ is an $\epsilon$-Nash equilibrium for the game $\Phi$ if it satisfies

\begin{equation}
Q_{i,\Phi}(\hat{\ve{a}}_i,\hat{\ve{a}}_{-i})\geq Q_{i,\Phi}(\ve{a}'_i,\hat{\ve{a}}_{-i})-\epsilon,~\forall \ve{a}'_i\in \mathcal{A}_i\backslash\hat{\ve{a}_i},~\forall{i\in \mathcal{N}}\label{eq17}
\end{equation}
where $\hat{\ve{a}}_{-i}$ is the equilibrium mixed strategy profile of the other aggregators $i'\in \mathcal{N}\backslash i,~\mathcal{A}_i$ is the set of all possible mixed strategy profiles of aggregator $i$ over $\mathcal{I}_i$, and $\epsilon>0$. Note that, in \eqref{eq17}, $Q_{i,\Phi}$ generalizes $Q_{i,\Phi}^{\text{EUT}}$ under expected utility theory and $Q_{i,\Phi}^{\text{PT}}$ under prospect theory.

\begin{remark}
Proposition 2 implies that finding a mixed strategy $\epsilon$-Nash equilibrium of the game $\Phi$ under expected utility theory or prospect theory leads to a mixed strategy subgame perfect $\epsilon$-Nash equilibrium \cite{Flesch} of the game $\Upsilon$.
\end{remark}

For the game $\Phi$, under each user behavioral model, an $\epsilon$-Nash equilibrium that is closely located to a mixed strategy Nash equilibrium is explored. To this end, the iterative algorithm proposed in \cite{Wang} is used where the algorithm was proven to converge to an $\epsilon$-Nash equilibrium close to a mixed strategy Nash equilibrium of a finite non-cooperative game under both expected utility theory and prospect theory. In a nutshell, the algorithm is given by

\begin{equation}
\ve{a}^{(k+1)}_i=\ve{a}^{(k)}_i+\frac{\beta}{k+1}(\ve{z}^{(k)}_i-\ve{a}^{(k)}_i),~0<\beta<1 \label{eq18}
\end{equation}
where $\beta$ is the inertia weight. Moreover, $\ve{z}^{(k)}_i=\big(z^{(k)}_i{(t_{i,1})},\dotsm,z^{(k)}_i(t_{i,\tilde{t}_i})\big)$ where

\begin{equation}
z_i^{(k)}(t_{i,t})=\begin{cases}
                 1,~\text{if}~t_{i,t}=\operatornamewithlimits{argmax}\limits_{t_i\in \mathcal{I}_i}~q_i(t_i,\ve{a}^{(k-1)}_{-i}), \\
                 0,~\text{otherwise}.
                   \end{cases}\label{eq19}
\end{equation}
Here, $q_i(t_i,\ve{a}^{(k-1)}_{-i})$ is the expected payoff of aggregator $i$ when they select the pure strategy $t_i$ for a given mixed strategy profile $\ve{a}^{(k-1)}_{-i}$ of their opponents $i'\in \mathcal{N}\backslash i$ at the iteration $k-1$. For prospect theory, $\ve{a}^{(k-1)}_{-i}$ considers the weighted probabilities of the aggregators $i'\in \mathcal{N}\backslash i$ at the iteration $k-1$. When the above algorithm converges, the $\epsilon$-Nash equilibrium with regard to $\ve{a}$ is found under expected utility theory and prospect theory. 

\section{Simulation Results}\label{sec:7}

\subsection{Simulation Setup}\label{sec:7_1}

To numerically examine the impacts of the EV charging competition among the EV aggregators with their ideal and non-ideal behavior, a system with five EV aggregators $(N=5)$ was considered. The EV charging scheduling time frame $\mathcal{T}$ spans from 8.00 AM to 4.00 PM and $T=16$ with $\Delta t=30~\text{min}$.

It was assumed that the EV fleet at each aggregator $i$ has 10 EVs, and EV chargers at each aggregator $i$ uses Level 2 charging. Level 2 charging is the primary approach used for EV charging at public places and typically uses EV charging rates between 3 kW and 20 kW \cite{Yilmaz}. For simulations in this paper, three types of EVs were considered, namely, Toyota Prius (3.8 kW,~4.4 kWh), Chevrolet Volt (3.8 kW,~16 kWh), and Nissan Leaf (3.3 kW,~24 kWh) \cite{Yilmaz}. The distributions of different types of EVs at each aggregator $i$ are given in Table~\ref{table 1}.

\begin{table}[b!]
\footnotesize
\caption{Number of different types of EVs available at each aggregator $i\in \mathcal{N}$}\vspace{5 mm}\label{table 1}
\centering
\begin{tabular}{|c|c|c|c|}
\hline
$i$ & \makecell*{Toyota Prius}&\makecell*{Chevrolet Volt} & \makecell*{Nissan Leaf} \\
\hline
1 &2& 3 & 5 \\
\hline
2 &2& 5 & 3\\
\hline
3 &3& 2 & 5 \\
\hline
4 &3& 5 & 2 \\
\hline
5 & 5&3 & 2 \\
\hline
\end{tabular}
\end{table}

Initial percentage SOC levels of the EVs controlled by each aggregator $i$ were randomly chosen between 0\% and 100\% of EVs' maximum energy storage capacities. It was assumed that all EVs should be charged to 100\% of their maximum energy storage capacities by the time $\mathcal{T}$. The charging efficiency $\eta_i$ of EV chargers controlled by each aggregator $i$ was assumed to be $0.864$ that is equivalent to the average Level 2 charging efficiency given in \cite{Forward}. Target energy demand profiles of each aggregator $i$ were generated as explained in Section~\ref{sec:3_2}.
For each aggregator $i$, $g_{i,t}$ was randomly chosen from the set $\{10,11,\dotsm,20\}$. Under these circumstances, Table~\ref{table 2} presents the possible EV charging start time strategy profiles $\mathcal{I}_i$ for the considered set of aggregators.

\begin{table}[t!]
\renewcommand{\arraystretch}{1.3}
\caption{EV charging start time strategy profiles $\mathcal{I}_i,~\forall i\in \mathcal{N}$}\vspace{5 mm}\label{table 2}
\centering
\begin{tabular}{|P{2.8 cm}|c|}
\hline
$i$ & \makecell*{EV Charging Start Time\\ Strategy Profiles $\mathcal{I}_i$} \\
\hline
1 &$\{1,2,\dotsm,5\}$ \\
\hline
2 &$\{1,2,\dotsm,7\}$\\
\hline
3 &$\{1,2,\dotsm,10\}$ \\
\hline
4 &$\{1,2,\dotsm,8\}$ \\
\hline
5 & $\{1,2,\dotsm,11\}$ \\
\hline
\end{tabular}
\end{table}

In simulations, $\ve{L}_b$ was assumed to be the aggregate energy demand of 200 residential facilities where the average energy demand profile between 8.00 AM and 4.00 PM is equivalent to the average energy demand profile of the Western Power Network in Australia between 8.00 AM and 4.00 PM in a Spring day \cite{WPNdata}. For grid pricing, $\phi_t=0.2$~AU cents/$\text{kWh}^2$ and $\delta_t=0.2$~AU cents/kWh at each time $t\in \mathcal{T}$ so that the peak unit energy price of the grid when all EVs at each aggregator $i$ are charged using their maximum charging power rates is equivalent to the peak usage domestic time-of-use tariff in \cite{origin}.

For the algorithm in \eqref{eq18}, initial probability distributions $\ve{a}_i^{(0)}$ were selected such that $\sum_{\mathcal{I}_i}a_i(t_i)=1$ for each aggregator $i$, and $\beta=0.7$. To compare results, an uncoordinated EV charging scenario was considered where all aggregators begin to charge their EV fleets from the time slot 1 using EVs' maximum charging power rates. The uncoordinated charging scenario uses the same energy cost models for the aggregators $\mathcal{N}$ in Section~\ref{sec:3}.

\subsection{Results and Discussion}\label{sec:7_2}

Fig.~\ref{fig1} shows the expected cost savings of each aggregator in $\mathcal{N}$ compared to the uncoordinated charging scenario under expected utility theory and prospect theory for two different $\alpha\in (0,1]$ values $(\alpha=0.1~\text{and}~\alpha=0.7)$. Here, it is assumed that $\alpha_i=\alpha,~\forall i\in \mathcal{N}$ where the probability weighting parameter $\alpha$ is applied according to \eqref{eq15}. It is important to note that when $\alpha=0.1$ aggregators become more subjective and non-ideal than when $\alpha=0.7$ because as $\alpha$ tends to $0$ from $1$, aggregators deviate further from the objective behavior assumed in expected utility theory. Table~\ref{table 3} presents the mixed strategy $\epsilon$-Nash equilibria obtained for the game $\Phi$ under expected utility theory and prospect theory.
\begin{figure}[t!]
\centering
\includegraphics[width=0.95\columnwidth]{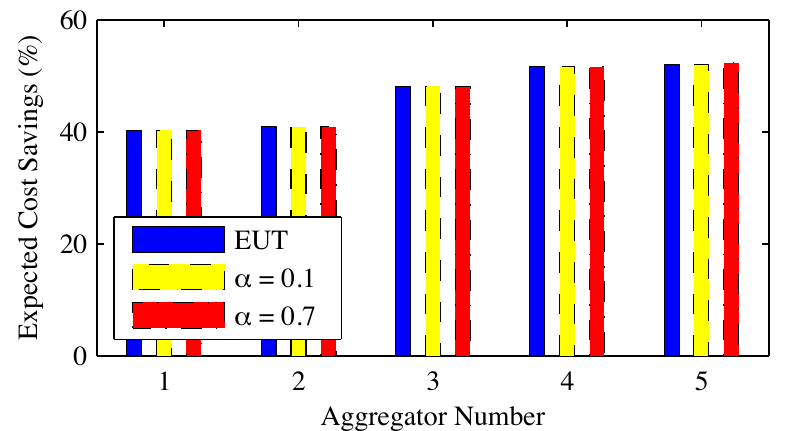}
\caption{Expected cost savings for the aggregators $\mathcal{N}$ under expected utility theory and prospect theory when $\alpha=0.1$ and $\alpha=0.7$.}
\label{fig1}
\end{figure}

\begin{table*}[t]
\renewcommand{\arraystretch}{1.3}
\caption{Percentage participation probabilities of the aggregators $i\in\mathcal{N}$ for $t_i\in \mathcal{I}_i$ under expected utility theory (EUT) and prospect theory (PT) when $\alpha = 0.7,~0.1$}\vspace{5 mm}
\label{table 3}
\centering
\begin{tabular}{|c|c|c|c|c|cIc|c|c|c|cIc|c|c|c|c|}
\hline
{}&\multicolumn{5}{cI}{\makecell*{EUT probabilities \\ i.e., when $\alpha=1$~(\%)}}&\multicolumn{5}{cI}{\makecell*{PT probabilities  \\when $\alpha=0.7$ (\%)}}&\multicolumn{5}{c|}{\makecell*{PT probabilities \\ when $\alpha=0.1$ (\%)}} \\
  \cline{2-16}
$t_i$ &\makecell*{$i=1$} &\makecell*{$i=2$} &\makecell*{$i=3$} & \makecell*{$i=4$} & \makecell*{$i=5$} & \makecell*{$i=1$}& \makecell*{$i=2$} & \makecell*{$i=3$} & \makecell*{$i=4$}& \makecell*{$i=5$}& \makecell*{$i=1$} & \makecell*{$i=2$} & \makecell*{$i=3$} & \makecell*{$i=4$}& \makecell*{$i=5$} \\
\hline
1 & 99.75 & 99.75 & 99.04  & 4.48 & 1.59 & 99.81 & 99.81 & 99.81 & 99.81 & 6.31 & 99.12 & 99.12 & 99.12 & 99.12 & 99.12\\
\hline
2 & 0.03 & 0.03 & 0.03 & 0.03 & 0.03 & 0.02 & 0.02 & 0.03 & 0.02 & 0.02 & 0.1 & 0.1 & 0.08 & 0.1 & 0.1\\
\hline
3 & 0.03 & 0.03 & 0.73 & 95.31 & 98.18 & 0.02 & 0.02 & 0.02 & 0.02 & 93.52 & 0.1 & 0.1 & 0.1 & 0.1 & 0.05\\
\hline
4 & 0.03 & 0.04 & 0.02 & 0.03 & 0.01 & 0.02 & 0.02 & 0.02 & 0.02 & 0.01 & 0.1 & 0.08 & 0.1 & 0.1 & 0.03\\
\hline
5 & 0.16 & 0.05 & 0.03 & 0.03 & 0.01 & 0.13 & 0.04 & 0.02 & 0.02 & 0.02 & 0.58 & 0.2 & 0.1 & 0.1 & 0.1\\
\hline
6 & - & 0.05 & 0.03 & 0.01 & 0.03 & - & 0.04 & 0.02 & 0.02 & 0.02 & - & 0.2 & 0.1 & 0.1 & 0.1\\
\hline
7& - & 0.05 & 0.03 & 0.03 & 0.03 & - & 0.05 & 0.02 & 0.02 & 0.02 & - & 0.2 & 0.1 & 0.1 & 0.1\\
\hline
8& - & - & 0.03 & 0.08 & 0.03 & - & - & 0.02 & 0.07 & 0.02 & - & - & 0.1 & 0.28 & 0.1\\
\hline
9& - & - & 0.03 & - & 0.03 & - & - & 0.02 & - & 0.02 & - & - & 0.1 & - & 0.1\\
\hline
10& - & - & 0.03 & - & 0.03 & - & - & 0.02 & - & 0.02 & - & - & 0.1 & - & 0.1\\
\hline
11 & - & - & - & - & 0.03 & - & - & - & - & 0.02 & - & - & - & - & 0.1\\
\hline

\end{tabular}
\end{table*}

According to the table, when aggregators have more subjective behavior with $\alpha=0.1$, the equilibrium probability distributions over $\mathcal{I}_i$ of each aggregator $i$ deviate from that of expected utility theory. In particular, when $\alpha=0.1$, the fourth and fifth aggregators prefer to participate from the time slot 1, whereas they prefer to participate from the time slot 3 under expected utility theory. When aggregators behave closer to the objective behavior by adopting $\alpha=0.7$, the fifth aggregator is more likely to start EV charging in the time slot 3 and the fourth aggregator prefers the time slot 1. Despite the changes in probabilistic choices of choosing an EV charging start time, Fig.~\ref{fig1} depicts that for each $\alpha$ value in prospect-theoretic analysis, the expected cost savings for the aggregators $\mathcal{N}$ remain almost same as the savings obtained under expected utility theory $(\alpha=1)$.

Fig.~\ref{fig2} illustrates the aggregators' expected EV charging grid loads in the time slot 1 under prospect theory with $\alpha=0.1$ and $\alpha=0.7$ compared to that of expected utility theory. Fig.~\ref{fig3} shows the temporal variation of the expected aggregate grid load after studying the EV charging competition under expected utility theory and prospect theory when $\alpha=0.1$.
\begin{figure}[t!]
\centering
\includegraphics[width=0.95\columnwidth]{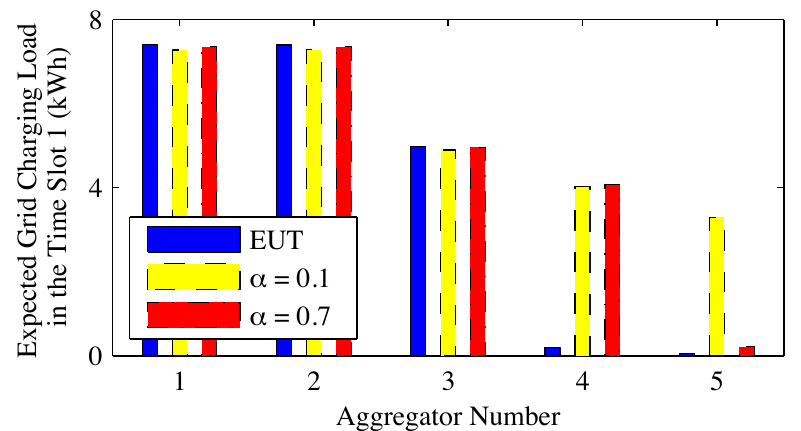}
\caption{Expected EV charging loads on the grid in the time slot 1 of the aggregators $\mathcal{N}$ under expected utility theory and prospect theory when $\alpha=0.1$ and $\alpha=0.7$.}
\label{fig2}
\end{figure}
\begin{figure}[t!]
\centering
\includegraphics[width=0.95\columnwidth]{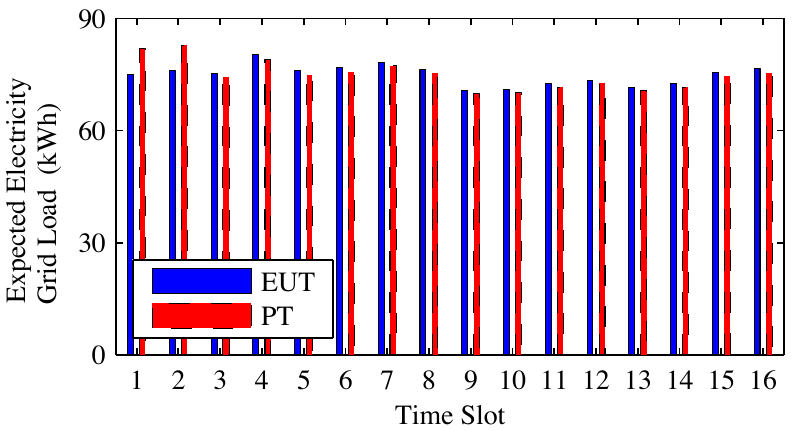}
\caption{Temporal variation of expected aggregate grid load under expected utility theory and prospect theory when $\alpha=0.1$.}
\label{fig3}
\end{figure}
Fig.~\ref{fig2} shows that when $\alpha=0.1$, the fifth aggregator incurs a significant EV charging load on the grid in the time slot 1 compared to their expected EV charging grid loads in expected utility theory and in prospect theory with $\alpha=0.7$. Similarly, when $\alpha=0.1$, the fourth aggregator also has a significant charging load in the time slot 1 compared to expected utility theory. Similar trends were observed for the time slot 2 as well. This is because both aggregators prefer to participate from the time slot 1 when $\alpha=0.1$, whereas, in expected utility theory, they prefer to participate from the time slot 3 (see Table~\ref{table 3}). As shown in Fig.~\ref{fig3}, the increase in EV charging loads of the fourth and fifth aggregators when $\alpha=0.1$ results in nearly 9\% higher load on the grid in each time slot (time slot 1 and 2) than expected utility theory.

Next, the influences of EV charging competition among the EV aggregators were studied across a range of possible $\alpha$ values. Here, $\alpha$ was varied in the range $(0,1]$. Fig.~\ref{fig4} illustrates the average expected cost savings of the aggregators $\mathcal{N}$ compared to the uncoordinated charging scenario with respect to changes in $\alpha$. Fig.~\ref{fig5} depicts the variations of expected peak-to-average ratio reductions compared to the uncoordinated charging scenario with varying $\alpha$. From the grid's perspective, a higher peak-to-average ratio reduction is preferred because it implies better peak load regulation compared to the uncoordinated EV charging case.
\begin{figure}[b!]
\centering
\includegraphics[width=0.95\columnwidth]{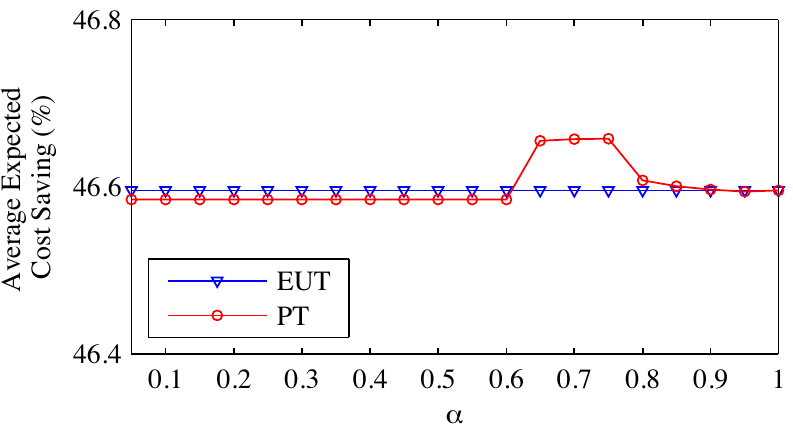}
\caption{Average of expected cost savings of the aggregators $\mathcal{N}$ with different $\alpha$.}
\label{fig4}
\end{figure}
\begin{figure}[t!]
\centering
\includegraphics[width=0.95\columnwidth]{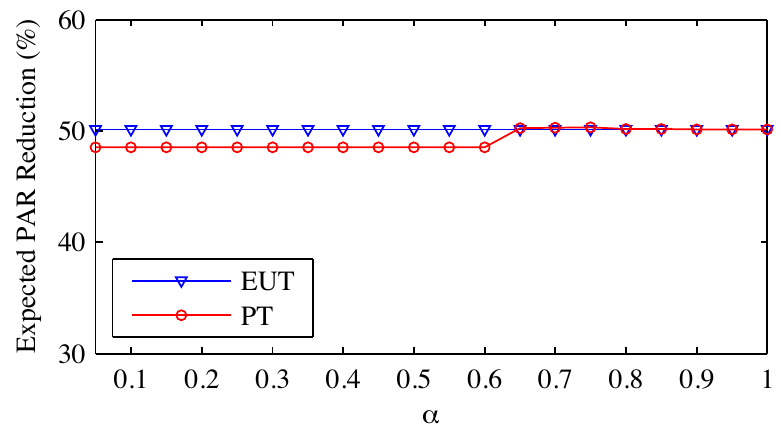}
\caption{Expected peak-to-average ratio (PAR) reduction with different $\alpha$.}
\label{fig5}
\end{figure}
According to Fig.~\ref{fig4}, when $0\leq \alpha \leq 0.6$, the average expected cost savings of the aggregators $\mathcal{N}$ are slightly lower than that obtained under expected utility theory. In particular, compared to the average expected cost saving under expected utility theory, this reduction is insignificant with only 0.01\%. When $0.65\leq\alpha \leq 0.75$, each aggregator in $\mathcal{N}$ receives a higher average expected cost saving with nearly 0.1\% increase. On the other hand, Fig.~\ref{fig5} shows that the expected peak-to-average ratio reductions remain nearly unchanged across the range of $\alpha$. When $0\leq\alpha\leq 0.6$, the peak-to-average ratio reductions are slightly lower than those achieved under expected utility theory. This is because, in this range of $\alpha$, the EV charging competition among the aggregators leads to higher peak loads on the grid than the peak grid load under expected utility theory, for example, as shown in Fig.~\ref{fig3}.

Finally, the impacts of the EV charging competition when each aggregator $i$ has different $\alpha$ values, i.e., $\alpha_i=\alpha_j$ for $\{i.j\}\in \mathcal{N},~i\neq j$, were investigated. To this end, it was considered $\ve{\alpha}=(0.7,~0.5,~0.9,~0.1,~0.3)$ as the matrix of $\alpha_i$ of the aggregators $\mathcal{N}$ under prospect theory. All other parameters are as specified for the previous simulation. Fig.~\ref{fig6} shows the expected cost savings for each aggregator in $\mathcal{N}$ compared to the uncoordinated EV charging scenario under expected utility theory and prospect theory.
\begin{figure}[t!]
\centering
\includegraphics[width=0.95\columnwidth]{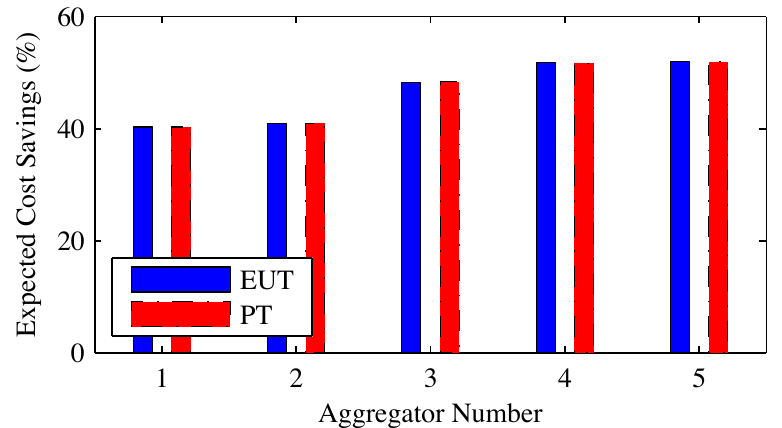}
\caption{Comparison of the expected cost savings of the aggregators $\mathcal{N}$ under expected utility theory and prospect theory when $\ve{\alpha}=(0.7,~0.5,~0.9,~0.1,~0,3)$.}
\label{fig6}
\end{figure}
Similar to the case where $\alpha_i=\alpha,~\forall i\in \mathcal{N}$, here the prospect-theoretic energy cost savings remain nearly the same as expected cost savings under expected utility theory. On the other hand, the expected peak-to-average ratio reduction of 50.16\% under expected utility theory increases slightly to 50.30\% reduction in the prospect-theoretic scenario with $\ve{\alpha}=(0.7,~0.5,~0.9,~0.1,~0.3)$.

\section{Conclusion}\label{sec:8}

This paper investigated impacts of non-ideal, subjective participating behavior of multiple electric vehicle (EV) aggregators, which might be run by car park managers, interacting in a coordinated EV charging competition. In the presented EV charging strategy, each aggregator minimizes their individual EV charging costs by selecting optimal EV charging start times and energy profiles. The EV charging competition among the aggregators was modeled by developing a two-stage non-cooperative game among the aggregators, which was studied under prospect theory to incorporate non-ideal participating actions of the aggregators. The non-cooperative EV charging game can obtain a subgame perfect $\epsilon$-Nash equilibrium when the game is played with either ideal, or non-ideal, participating actions of the aggregators. Through numerical simulations, we have shown that the benefits of the coordinated EV charging strategy, in terms of EV charging energy cost reductions and peak load regulation, are significantly resilient to non-ideal participating actions taken by the aggregators.  

Future work could focus on extending the analysis in this work to investigate the charging system by including uncertainties that arise with stochastic behavior of EVs using insights from Bayesian game theory \cite{gametheoryessentials}. Further study could incorporate both grid-to-vehicle and vehicle-to-grid operations with demand-side management so that EV storage devices can be utilized as distributed energy resources for energy management. Moreover, it would also be interesting to investigate the effects of non-ideal consumer behavior on an EV charging management framework that incorporates renewable energy sources.


 \newcommand{\noop}[1]{}


\begin{IEEEbiography}
{Chathurika P. Mediwaththe} (S'12) received the B.Sc. degree (Hons.) in electrical and electronic engineering from the University of Peradeniya, Sri Lanka. She completed the PhD degree in electrical engineering at the University of New South Wales, Sydney, NSW, Australia in 2017. From 2013-2017, she was a research student with Data61-CSIRO (previously NICTA), Sydney, NSW, Australia. She is currently a research fellow at the Australian National University, Canberra, ACT, Australia. Her current research interests include electricity demand-side management, smart grids, decision making (game theory and optimization) for resource allocation in distributed networks, and machine learning. 
\end{IEEEbiography}

\begin{IEEEbiography}
{David B. Smith} (S'01 - M'04) received the B.E. degree in electrical engineering from the University of New South Wales, Sydney, NSW, Australia, in 1997, and the M.E. (research) and Ph.D. degrees in telecommunications engineering from the University of Technology, Sydney, Ultimo, NSW, Australia, in 2001 and 2004, respectively. 

Since 2004, he was with National Information and Communications Technology Australia (NICTA, incorporated into Data61 of CSIRO in 2016), and the Australian National University (ANU), Canberra, ACT, Australia, where he is currently a Senior Research Scientist with Data61 CSIRO, and an Adjunct Fellow with ANU. He has a variety of industry experience in electrical and telecommunications engineering. His current research interests include wireless body area networks, game theory for distributed networks, mesh networks, disaster tolerant networks, radio propagation, 5G networks, antenna design, distributed optimization for smart grid and privacy for networks.  He has published over 100 technical refereed papers. He has made various contributions to IEEE standardisation activity. He has served (or is serving) on the technical program committees of several leading international conferences in the fields of communications and networks. Dr. Smith was the recipient of four conference Best Paper Awards.
\end{IEEEbiography}

\end{document}